\documentclass{article}
\usepackage{amsthm,authblk,fullpage,amssymb}

\usepackage[utf8]{inputenc}
\usepackage{tikz}
\usepackage{comment}
\usetikzlibrary{calc,decorations, decorations.pathreplacing,decorations.pathmorphing,arrows}
\usepackage{amssymb}
\usepackage{enumerate}
\usepackage{mathtools}

 \usepackage[hidelinks]{hyperref}
  \usepackage[capitalise]{cleveref}
\usepackage{breakurl}

 \theoremstyle{plain}
  \newtheorem{theorem}{Theorem}[section]
  \newtheorem{lemma}[theorem]{Lemma}  
  \newtheorem{observation}[theorem]{Observation}  
  \newtheorem{corollary}[theorem]{Corollary}  
  \newtheorem{fact}[theorem]{Fact}
  \crefname{figure}{Figure}{Figures}
  \theoremstyle{definition}
  \newtheorem{hypothesis}[theorem]{Hypothesis}
  
  \theoremstyle{remark}
  \newtheorem{remark}[theorem]{Remark}
  \newtheorem{example}[theorem]{Example}
  
\crefname{hypothesis}{Hypothesis}{Hypotheses}

\renewcommand{\P}{\mathcal{P}}
\newcommand{\ww}{\mathbf{w}}
\newcommand{\pp}{\mathbf{p}}
\newcommand{\qq}{\mathbf{q}}
\newcommand{\QQ}{\mathbf{Q}}

\newcommand{\NN}{\mathbf{N}}
\newcommand{\SQ}{\mathit{SQ}}
\renewcommand{\Alph}{\mathrm{Alph}}
\newcommand{\Fib}{\mathit{Fib}}
\newcommand{\lleft}{\psi}

\newcommand{\SQABEL}{\SQ_{\mathrm{Abel}}}
\newcommand{\SQPABEL}{\SQ'_{\mathrm{Abel}}}
\newcommand{\SQPARAM}{\SQ_{\mathrm{param}}}
\newcommand{\SQPPARAM}{\SQ'_{\mathrm{param}}}
\newcommand{\SQOP}{\SQ_{\mathrm{op}}}
\newcommand{\SQPOP}{\SQ'_{\mathrm{op}}}
\newcommand{\hasha}[2]{|#1|_{#2}}
\renewcommand{\L}{\mathcal{L}}
\DeclareMathOperator{\ind}{index}
\DeclareMathOperator{\id}{id}

\DeclarePairedDelimiter{\floor}{\lfloor}{\rfloor}

\title{Maximum Number of Distinct and Nonequivalent Nonstandard Squares in a Word}

\author{Tomasz Kociumaka\footnote{Supported by Polish budget funds for science in 2013-2017 as a research project under the 'Diamond Grant' program,
  grant no 0179/DIA/2013/42.}}
\author{Jakub Radoszewski\footnote{The author is a Newton Fellow at King's College London.} \footnote{Supported by the Polish Ministry of Science and Higher Education under the `Iuventus Plus' program in 2015-2016 grant no 0392/IP3/2015/73.}}
\author{Wojciech Rytter\footnote{Supported by the Polish National Science Center, grant no 2014/13/B/ST6/00770.}}
\author{Tomasz Waleń$^\ddagger$}
\affil{Institute of Informatics, University of Warsaw\\
    \texttt{[kociumaka,jrad,rytter,walen]@mimuw.edu.pl}}

\date{}

\begin{document}
\maketitle

\begin{abstract}
  The combinatorics of squares in a word depends on how the equivalence of halves of the square is defined.
  We consider Abelian squares, parameterized squares, and order-preserving squares.
  The word $uv$ is an Abelian (parameterized, order-preserving) square if $u$ and $v$ are
  equivalent in the Abelian (parameterized, order-preserving) sense.
  The maximum number of ordinary squares in a word is known to be
  asymptotically linear, but the exact bound is still investigated.
  We present several results on the maximum number of distinct squares
  for nonstandard subword equivalence relations.
  Let $\SQABEL(n,\sigma)$ and $\SQPABEL(n,\sigma)$ denote the maximum number of Abelian
  squares in a word of length $n$ over an alphabet of size $\sigma$,
  which are distinct as words and which are nonequivalent in the Abelian sense, respectively.
  For $\sigma\ge 2$ we prove that $\SQABEL(n,\sigma)=\Theta(n^2)$, $\SQPABEL(n,\sigma)=\Omega(n^{3/2})$
  and $\SQPABEL(n,\sigma) = O(n^{11/6})$.
  We also give linear bounds for parameterized and order-preserving squares
  for alphabets of constant size: $\SQPARAM(n,O(1))=\Theta(n)$, $\SQOP(n,O(1))=\Theta(n)$.
  The upper bounds have quadratic dependence on the alphabet size for order-preserving
  squares and exponential dependence for parameterized squares. 
  
  As a side result we construct infinite words over the smallest alphabet
  which avoid nontrivial order-preserving squares and nontrivial parameterized cubes
  (nontrivial parameterized squares cannot be avoided in an infinite word).

  A preliminary version of this paper was published at DLT 2014 [LNCS vol. 8633. Springer, pp. 216--226, 2014].
  In this full version we improve or extend the bounds on all three kinds of squares.
\end{abstract}

\section{Introduction}
  Repetitions in words are a fundamental topic in combinatorics on words \cite{Karhumaki}.
  They are widely used in many fields, such as pattern matching, automata theory,
  formal language theory, data compression, molecular biology, etc.
  Squares, that is, words of the form $uu$, are the basic and one of the most commonly studied types of repetitions.
  An example of an infinite square-free word over a ternary alphabet, given by Thue \cite{Thue},
  is considered to be the foundation of combinatorics on words.

  If we allow other equivalence relations on words, several generalizations of the notion of square
  can be obtained.
  One such generalization are Abelian squares, that is, words of the form $uv$ where
  the multisets of symbols of $u$ and $v$ are the same.
  Abelian squares were first studied by Erd\H{o}s \cite{Erdos}, who posed a question on
  the smallest alphabet size for which there exists an infinite Abelian-square-free word.
  The first example of such a word over a finite alphabet was given by Evdokimov~\cite{evdokimov}.
  Later the alphabet size was improved to five by Pleasants~\cite{Pleasants}
  and finally an optimal example over a four-letter alphabet was shown by
  Ker\"anen~\cite{DBLP:conf/icalp/Keranen92}.

  In this paper we consider Abelian squares and introduce squares based
  on two other known equivalence relations on words.
  The first is parameterized equivalence \cite{DBLP:journals/jcss/Baker96}, in which two words
  $u$, $v$ of length $n$ over alphabets $\Alph(u)$ and $\Alph(v)$ are considered equal
  if one can find a bijection $f : \Alph(u) \to \Alph(v)$
  such that $v[i]=f(u[i])$ for all $i=1,\ldots,n$.
  The second model, order-preserving equivalence \cite{DBLP:journals/ipl/KubicaKRRW13,DBLP:journals/tcs/KimEFHIPPT14},
  assumes that the alphabets are ordered.
  Two words $u$, $v$ of the same length are considered equivalent in this model if they are
  equal in the parameterized sense with $f$ being a strictly increasing bijection.
  We define a parameterized square and an order-preserving square as a concatenation
  of two words that are equivalent in the parameterized and in the order-preserving sense, respectively.
  Another recently studied model, which we do not consider in our work, however,
  is $k$-Abelian equivalence \cite{DBLP:journals/tcs/HuovaKS12}.
  It lies in between standard equality and Abelian equivalence.
  The nonstandard types of squares can be viewed as a part of nonstandard stringology;
  see \cite{DBLP:conf/cpm/Muthukrishnan95,Muthukrishnan:1994:NSA:195058.195457}.
  Algorithms for computing Abelian squares and order-preserving squares were recently presented
  in \cite{MACIS2015} and \cite{Crochemore2015}, respectively.

  \begin{example}
    Consider the alphabet $\Sigma=\{1,2,3,4\}$ with the natural order.
    Then $1213\,1213$ is a square,
    $1213\,3112$ is an Abelian square,
    $1213\,4142$ is a parameterized square,
    and $1213\,1314$ is an order-preserving square over $\Sigma$.
  \end{example}

  An important combinatorial fact about ordinary squares is that the maximum number
  of distinct squares in a word of length $n$ is linear in terms of $n$.
  Actually this number has recently been proved to be at most $\frac{11}{6}n$~\cite{DBLP:journals/dam/DezaFT15},
  improving upon an earlier bound of $2n-\Theta(\log n)$
  \cite{fraenkel-simpson,DBLP:journals/jct/Ilie05,DBLP:journals/tcs/Ilie07}.
  This bound has found applications in several text algorithms
  \cite{DBLP:journals/tcs/CrochemoreIR09} including two different linear-time algorithms
  computing all distinct squares \cite{DBLP:journals/jcss/GusfieldS04,Extracting_TCS}.
  A recent result shows that the maximum number of distinct squares in a labeled tree is
  asymptotically $\Theta(n^{4/3})$ \cite{DBLP:conf/cpm/CrochemoreIKKRRTW12}.
  Also some facts about counting distinct squares in partial words are known
  \cite{DBLP:conf/dlt/Blanchet-SadriJM12,DBLP:journals/actaC/Blanchet-SadriMS09}.
  In this paper we attempt the same type of combinatorial analysis for nonstandard squares.
  In turns out that the results that we obtain depend heavily on which squares
  we consider distinct.

  Let $\SQABEL(n,\sigma)$, $\SQPARAM(n,\sigma)$, and $\SQOP(n,\sigma)$ denote respectively the maximum number of
  Abelian, parameterized, and order-preserving squares in a word of length $n$ over
  an alphabet of size $\sigma$ which are \emph{distinct} as words.
  Moreover, let $\SQPABEL(n,\sigma)$, $\SQPPARAM(n,\sigma)$, and $\SQPOP(n,\sigma)$ denote the maximum number of
  Abelian, parameterized, and order-preserving squares in a word of length $n$ over
  an alphabet of size $\sigma$ which are \emph{nonequivalent} in the Abelian, parameterized,  and order-preserving sense, respectively.
  We also use analogous notation, e.g., $\SQABEL(w)$, $\SQPABEL(w)$, for an arbitrary word $w$.

  \begin{example}
    Consider a Fibonacci word\footnote{%
      Fibonacci words are defined as: $\Fib_0=0$, $\Fib_1=01$, $\Fib_k = \Fib_{k-1} \Fib_{k-2}$ for $k \ge 2$.
    }  $\Fib_5=0100101001001$.
    It contains 5 Abelian squares of length 6:
    \[010\,010,\ 001\,010,\ 010\,100,\ 100\,100,\ \mbox{and}\ 001\,001,\]
    which are all distinct as words but are Abelian-equivalent.
    In total, $\Fib_5$ contains 13 distinct subwords which are Abelian squares.
    Hence, $\SQABEL(\Fib_5)=13$.
    On the other hand, $\Fib_5$ contains only 5 Abelian-nonequivalent squares, with sample representatives:
    \[0\,0,\ 01\,01,\ 001\,010,\ 10010\,10010,\ \mbox{and}\ 010010\,100100.\]
    Hence, $\SQPABEL(\Fib_5)=5$.
    The value $\SQ'$ is usually much smaller than $\SQ$, e.g., for $\Fib_{14}$ of length $987$,
    $\SQPABEL(\Fib_{14})=490$ and
    $\SQABEL(\Fib_{14})=57796$.
    In general, one can show that $\SQPABEL(\Fib_k)=O(|\Fib_k|)$.
    Abelian repetitions in Fibonacci words and Sturmian words were already studied in \cite{DBLP:conf/dlt/FiciLLLMP13}.
  \end{example}

  \subsection*{Our main results}
  \begin{enumerate}[1.]
    \item $\SQPABEL(n,\sigma)=\Omega(n^{3/2})$ and $\SQPABEL(n,\sigma) \le n^{11/6}$ for $\sigma \ge 2$;
    \item $\SQOP(n,\sigma)\le (\binom{\sigma}{2}+\frac{11}{6})n$ and therefore $\SQPOP(n,\sigma)=(\binom{\sigma}{2}+\frac{11}{6})n$;
    \item $\SQPARAM(n,\sigma)\le 2(\sigma!)^2n$ 
      and $\SQPPARAM(n,\sigma)\le 2\sigma!n$.
  \end{enumerate}

  \subsection*{Structure of the paper}
  The two following sections are devoted to bounds on Abelian squares.
  In \cref{sec:lower} we first show a simple example which implies $\SQABEL(n,\sigma)=\Theta(n^2)$ for $\sigma \ge 2$.
  Next we construct a family of binary words that gives the lower bound $\SQPABEL(n,\sigma) = \Omega(n^{1.5})$ for $\sigma \ge 2$.
  In \cref{sec:upper} we present upper bounds related to $\SQPABEL(n,\sigma)$:
  we use a result from additive combinatorics to derive a general $n^{11/6}$ upper bound,
  and we prove an upper bound of $O(nm)$ which holds in the case that the number of blocks
  of the same letter (i.e., the size of the run-length encoding of the word) is bounded by $m$.  

  In the next two sections upper bounds for the number of order-preserving and parameterized squares,
  respectively, are presented in the case of a small alphabet.

  The final \cref{sec:infinite} can be viewed as an extension of the
  works of Thue \cite{Thue}, Evdokimov~\cite{evdokimov}, Pleasants~\cite{Pleasants}, and Ker\"anen~\cite{DBLP:conf/icalp/Keranen92}
  on infinite square-free and Abelian-square-free words into the parameterized and order-preserving equivalence.
  As no square-free word of length larger than 1 exists for these two models of equivalence,
  we consider words avoiding \emph{nontrivial} squares, of length larger than 2.
  We present an infinite word over the minimum-size (ternary) alphabet
  avoiding nontrivial order-preserving squares.
  We also prove that there is no infinite word avoiding nontrivial parameterized squares,
  but there is one avoiding nontrivial parameterized cubes, that is,
  parameterized cubes of length greater than 3.

  \subsection*{Preliminary notions}

  By $\Sigma^*$ we denote the set of finite words over the alphabet $\Sigma$ and by
  $\Sigma^n$ we denote the subset of $\Sigma^*$ containing words of length $n$.
  For a word $w=w[1]\cdots w[n]$ we denote $|w|=n$ and $\Alph(w)$ as the set of letters present in $w$.
  A \emph{subword} of $w$ is a word of the form $u=w[i] \cdots w[j]$ for $1\le i\le j \le |w|$.
  By $w[i..j]$ we denote the \emph{occurrence} of $u$ at position $i$, called a \emph{fragment}
  of $w$. 
  A fragment is said to be \emph{uniform} if all its letters are equal.
  A \emph{block} (also known as a \emph{run}) in a word is a maximal uniform fragment, i.e.,
  a uniform fragment which cannot be extended neither to the left nor to the right.
  
  For a word $w$ and a letter $c$ we denote the number of occurrences of $c$ in $w$ by $\hasha{w}{c}$.
  The \emph{Parikh vector} of a word $w$ over an ordered alphabet $\Sigma=\{0,\ldots,\sigma-1\}$
  is defined as $\P(w)=(\hasha{w}{0},\ldots, \hasha{w}{\sigma-1}).$
  Note that Parikh vectors belong to $\mathbb{Z}_{\ge 0}^\sigma$.

  \section{Lower Bounds for Abelian Squares}\label{sec:lower}
  Let us start with a simple bound for $\SQABEL(n,\sigma)$.
  A different proof of the following fact was given independently by Fici~\cite{Fici}.

  \begin{fact}{\label{fct:sq-ab-lower-bound}}
    $\SQABEL(n,\sigma)=\Theta(n^2)$ for $\sigma \ge 2$.
  \end{fact}
  \begin{proof}
    Consider the word $u_k=0^k 1 0^k 1 0^{2k}$ of length $4k+2$.
    It contains $\Theta(k^2)$ Abelian squares of the form
    $0^a 10^b\ 0^{k-b} 1 0^{a+2b-k}$ for all $a,b\in \mathbb{Z}_{\ge 0}$ such that $a,b\le k$ and $a+2b\ge k.$
    Thus we obtain $\SQABEL(n,2)=\Theta(n^2)$ for $n=4k+2$.
    If $n \bmod 4 \ne 2$, we pick the longest word $u_k$ such that $|u_k| \le n$
    and extend it with $n-|u_k| \le 3$ zeros.
  \end{proof}

  In the preliminary version~\cite{DBLP:conf/dlt/KociumakaRRW14} we showed that $\SQPABEL(n,2) = \Omega(n^{3/2}/\log n)$.
  The family of words used in that construction was $01 0^21^2 \ldots 0^k1^k$.
  Here, we prove that $\SQPABEL(n,2) = \Omega(n^{3/2})$.
  Our lower-bound family of words is
  \[\ww_k = (0^{k}1^{k})^{3k} (0^{k+1}1^{k+1})^{k}.\]
  
  We say that $(r,\ell)$ is \emph{a square vector} in $w$
  if there exists an Abelian square $u_1u_2$ in $w$ such that
  $\P(u_1)=\P(u_2)=(r,\ell)$.
  Now, $\SQPABEL(n,2)$ can be expressed as the maximum number of different square vectors in a binary word of length~$n$.

  In our construction \emph{balanced} Abelian squares
  and \emph{balanced} square vectors play a crucial role.
  A vector $(r,\ell)$ is called balanced if $r=\ell$, and a word $w$ is called
  balanced if its Parikh vector is balanced.
  Abelian squares and square vectors are called unbalanced if they are not balanced.

  In what follows, we identify $\Theta(k^2)$ distinct balanced Abelian squares in $\ww_k$ and extend them 
  obtaining further $\Theta(k)$ (unbalanced) Abelian squares for each balanced square vector $(\ell,\ell)$.
  Here, $\ell$ shall be an arbitrary element of the following set:
  \[
    \NN_k = \{ik+j(k+1)\,:\, 0 \le i \le k,\, 0 \le j \le k-1\}.
  \]

  \begin{example}
    $\NN_3 = \{0,4,8,\,3,7,11,\,6,10,14,\,9,13,17\}$.
  \end{example}
  
  \begin{observation}\label{obs:S}
    $|\NN_k| = k(k+1)$.
  \end{observation}
  \begin{proof}
    Let $i,i' \in \{0,\ldots,k\}$ and $j,j' \in \{0,\ldots,k-1\}$.
    Suppose that \[ik+j(k+1) = \ell = i'k+j'(k+1).\]
    We then have $j \equiv \ell \equiv j' \pmod{k}$ and
    $-i \equiv \ell \equiv -i' \pmod{k+1}$, i.e., $i=i'$ and $j=j'$.
    Consequently, $|\NN_k|=k(k+1)$
    as claimed. 
  \end{proof}

  \noindent
  Denote
  \[
    \qq_{i,j} = 1^j(0^k1^k)^{2i+j}(0^{k+1}1^{k+1})^j0^j
  \]
  and $\QQ_k = \{\qq_{i,j}\,:\, 0 \le i \le k,\, 0 \le j \le k-1\}$.
  
  \begin{lemma}\label{lem:balanced}
  Every $\qq_{i,j}\in \QQ_k$ is a balanced Abelian square corresponding to square vector $(\ell,\ell)$ with $\ell=ik+j(k+1)$.
  Moreover, $\qq_{i,j}$ occurs in $\ww_k$ at position $6k^2-4ik - (2k+1)j+1$.
  \end{lemma}
  \begin{proof}
  First, note that 
  \[|\qq_{i,j}|_0=(2i+j)k+j(k+1)+j=2ik+2j(k+1)=2\ell\]
  and similarly $|\qq_{i,j}|_1=j+(2i+j)k+j(k+1)=2\ell$. Also,
  observe that $\qq_{i,j}$ has a prefix
  \[\pp_{i,j} = 1^j (0^k1^k)^{i+j}0^j\] such
  that $|\pp_{i,j}|_0 = (i+j)k+j = ik+j(k+1)=\ell$ and $|\pp_{i,j}|_1 = j+(i+j)k=\ell$. Thus,
  $\qq_{i,j}$ is indeed a balanced Abelian square with square vector $(\ell,\ell)$.
  
  Finally, observe that $1^j(0^k1^k)^{2i+j}$ is a suffix of $(0^k1^k)^{3k}$ (since $2i+j\le 2k+k-1<3k$)
  and $(0^{k+1}1^{k+1})^j 0^j$ is a prefix of $(0^{k+1}1^{k+1})^{k}$ (since $j<k$).
  Consequently, $\qq_{i,j}$ occurs in $\ww_k$ at position
  \[1+|(0^k1^k)^{3k}|-|1^j(0^k1^k)^{2i+j}|=6k^2-4ik-(2k+1)j+1,\]
  as claimed.
  \end{proof}

  \begin{figure}[b]
    \centering
    \begin{tikzpicture}[xscale=0.3]
  \foreach \x in {-11,-7,-3,1,5,9}{
    \begin{scope}[xshift=\x cm]
    \foreach \y in {0,1} \draw[xshift=\y cm] node {\footnotesize 0};
    \foreach \y in {0,1} \draw[xshift=2cm,xshift=\y cm] node {\footnotesize 1};
    \end{scope}
  }
  \begin{scope}[xshift=12.5cm]
  \foreach \x in {1,7}{
    \begin{scope}[xshift=\x cm]
    \foreach \y in {0,1,2} \draw[xshift=\y cm] node {\footnotesize 0};
    \foreach \y in {0,1,2} \draw[xshift=3cm,xshift=\y cm] node {\footnotesize 1};
    \end{scope}
  }
  \end{scope}
  
  \draw[very thick] (5.5,0.2) -- (5.5,0.4) node[above] {\footnotesize $\qq_{2,1}$} -- (20,0.4) -- (20,0.2);
  \draw[very thick] (-8.5,0.2) -- (-8.5,0.4) -- (5.5,0.4);

  \draw (6.5,-0.2) -- (6.5,-0.4) -- (22,-0.4) -- (22,-0.2);
  \draw (-8.5,-0.2) -- (-8.5,-0.4) -- (6.5,-0.4);
\end{tikzpicture}
    \caption{
      Consider the word $\ww_2$ and its subword $\qq_{2,1}$.
      This subword corresponds to a balanced Abelian square with square vector
      $\frac{1}{2}\P(\qq_{2,1}) = (7,7)$ (in bold).
      The first half of the Abelian square is followed with $0$ and the second half with $0^2$.
      Hence, if we extend each half by one position, we obtain an unbalanced Abelian square
      with square vector $(8,7)$.
    }\label{fig:lb}
  \end{figure}
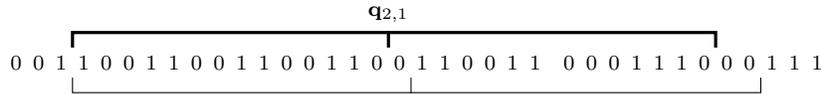

  An illustration of the proof of \cref{lem:balanced} is shown in \cref{fig:lb}.
  This figure also provides some intuition on how to obtain unbalanced Abelian squares
  from the balanced Abelian squares that we identified using this lemma.
  
  \begin{lemma}\label{lem:unbalanced}
  For each $\ell\in \NN_k$ the word $\ww_k$ contains at least $\frac{k+1}{2}$
  square vectors of the form $(r,\ell)$ for some integer $r$.
  \end{lemma}
  \begin{proof}
  By \cref{lem:balanced}, there exists $\qq_{i,j}\in \QQ_k$ whose square vector is $(\ell,\ell)$.
  Moreover, it occurs in $\ww_k$ at position $a := 6k^2-4ik-(2k+1)j+1$.
  
  We shall prove that for each $r\in \{\ell-\floor{\frac{j}{2}},\ldots,\ell+\floor{\frac{k+1-j}{2}}\}$, 
  there is an Abelian square with square vector $(r,\ell)$ occurring in $\ww_k$ at position $a$; see also \cref{fig:unbalanced}.
  
  Note that $\ww_k[a..a+4\ell-1]=\qq_{i,j}$ ends with $0^j$ and is followed by $0^{k+1-j}$,
  while its first half ends with $0^j$ (with $0^{\floor{\frac{j}{2}}}$ in particular) and is followed by $0^{k-j}$
  (by $0^{\floor{\frac{k+1-j}{2}}}$ in particular, because $j\le k-1$). 
  Consequently, we have
  \[\P(\ww_k[a..a+\ell+r-1])=\P(\ww_k[a..a+2\ell-1])+(r-\ell,0)=(r,\ell)\]
  and
  \[\P(\ww_k[a..a+2(\ell+r)-1])=\P(\ww_k[a..a+4\ell-1])+(2(r-\ell),0)=(2r,2\ell).\]
  Therefore, there are at least $\floor{\frac{j}{2}}+\floor{\frac{k+1-j}{2}}+1\ge \frac{k+1}{2}$
  square vectors of the claimed form in $\ww_k$.  
  \end{proof}

  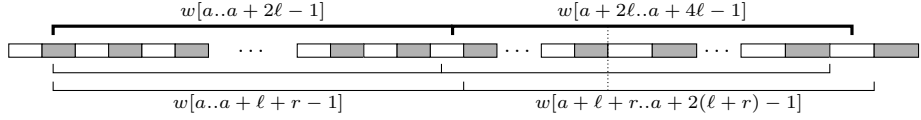
\begin{figure}[t]
    \centering
    \def\blockH{0.6}
\newcommand{\drawBlock}[4]{
  \draw[fill=white] (#1,0)--(#1/2+#2/2,0)--(#1/2+#2/2,\blockH)--(#1,\blockH)--cycle;
  \draw[fill=white!70!black] (#1/2+#2/2,0)--(#2,0)--(#2,\blockH)--(#1/2+#2/2,\blockH)--cycle;
  \node at (#1+#2/4-#1/4,\blockH) [above] {\small #3};
  \node at (#1/2+#2/2+#2/4-#1/4,\blockH) [above] {\small #4};
}
\begin{tikzpicture}[scale=0.295]

\drawBlock{-9}{-6}{}{}
\drawBlock{-6}{-3}{}{}
\drawBlock{-3}{0}{}{}
\drawBlock{4}{7}{}{}
\drawBlock{7}{10}{}{}
\drawBlock{10}{13}{}{}
\drawBlock{15}{18}{}{}
\drawBlock{18}{22}{}{}
\drawBlock{24}{28}{}{}
\drawBlock{28}{32}{}{}

\draw (2,0.3) node {\footnotesize \dots};
\draw (14,0.3) node {\footnotesize \dots};
\draw (23,0.3) node {\footnotesize \dots};

\draw[densely dotted] (18,-1.5) -- (18,1.5);

\draw[yshift=0.2cm,very thick,xshift=1cm] (10,0.8) -- (10,1.2) -- node[above=-2]{\scriptsize{$w[a+2\ell..a+4\ell-1]$}} (28,1.2) -- (28,0.8);
\draw[yshift=0.2cm,very thick,xshift=-17cm] (10,0.8) -- (10,1.2) -- node[above=-2]{\scriptsize{$w[a..a+2\ell-1]$}} (28,1.2);

\draw[xshift=1cm] (-8,-0.3) -- (-8,-0.7) -- (9.5,-0.7);
\draw[xshift=18.5cm] (-8,-0.3) -- (-8,-0.7) -- (9.5,-0.7) -- (9.5,-0.3);

\
\draw[yshift=-0.8cm,xshift=1cm] (-8,-0.3) -- (-8,-0.7) -- node[below=-2]{\scriptsize{$w[a..a+\ell+r-1]$}} (10.5,-0.7);
\draw[yshift=-0.8cm,xshift=19.5cm] (-8,-0.3) -- (-8,-0.7) -- node[below=-2]{\scriptsize{$w[a+\ell+r..a+2(\ell+r)-1]$}}  (10.5,-0.7) -- (10.5,-0.3);

\end{tikzpicture}
    \caption{\label{fig:unbalanced}
      A schematic illustration of the proof of \cref{lem:unbalanced}.
      Light rectangles represent zeroes and dark rectangles represent ones.
    }
  \end{figure}

  \begin{theorem}
    $\SQPABEL(n,\sigma)= \Omega(n^{3/2})$ for each $\sigma \ge 2$.
  \end{theorem}
  \begin{proof}
    We constructed a family of binary words $\ww_k$ together with the sets $\NN_k$.
    Note that 
    \[|\ww_k| = 8k^2+2k=O(k^2)\]
    and, by \cref{obs:S},
    \[|\NN_k|=k(k+1)=\Omega(k^2).\]
    By \cref{lem:unbalanced}, the number of distinct square vectors in $\ww_k$ is at least
    \[|\SQPABEL(\ww_k)|\ge \tfrac{k+1}{2}|\NN_k| = \Omega(k^3)=\Omega(|\ww_k|^{3/2}).\]
    This completes the lower bound proof for $n=|\ww_k|$.
    For other lengths $n$ we pick the longest word $\ww_k$ such that $|\ww_k| \le n$
    and append it with ones.
  \end{proof}

  \section{Upper Bounds for Abelian Squares}\label{sec:upper}
  Let us start with an upper bound of $n^{11/6}$ on the number
  of nonequivalent Abelian squares using the following result from additive combinatorics.
  Recall that an Abelian group $(Z,+)$ is called \emph{torsion-free} if $nz=0$ for $n\in \mathbb{Z}$ and $z\in Z$ implies $n=0$ or $z=0$.
  \begin{lemma}[Katz \& Tao \cite{katz1999bounds}]\label{lem:tao}
  Let $(Z,+)$ be a torsion-free Abelian group, let $A,B$ be subsets of $Z$, and let $G \subseteq A\times B$.
  If \[\max(|A|,|B|,|\{a+b : (a,b)\in G\}|)\le N,\] then
  \[|\{a-b : (a,b)\in G\}| \le N^{11/6}.\]
  \end{lemma}

  \begin{observation}
    The set $\mathbb{Z}^{\sigma}$ (containing all Parikh vectors) with addition is
    a torsion-free Abelian group.
  \end{observation}
  
  \begin{theorem}
  $\SQPABEL(n,\sigma)\le (n-1)^{11/6}$ for each $\sigma\ge 1$ and $n\ge 1$.
  \end{theorem}
  \begin{proof}
  Consider a word $w \in \Sigma^n$ where $\Sigma=\{0,\ldots,\sigma-1\}$ and a torsion-free Abelian group $Z=\mathbb{Z}^\sigma$.
  For $0\le i \le n$ let $\P_i = \P(w[1..i])\in \mathbb{Z}^\sigma$ be the Parikh
  vector of the $i$-th prefix of $w$. We set 
  \[A=\{\P_2,\P_3,\ldots,\P_n\},\ B=\{\P_0,\P_1,\ldots,\P_{n-2}\}\]
  and
  \[G=\{(\P_{j},\P_{i-1}) : w[i..j]\text{ is an Abelian square}\}.\]
  
  Note that $G\subseteq A\times B$ because every Abelian square has length at least 2.
  Moreover, $w[i..j]$ is an Abelian square if and only if $\P_{i-1}+\P_{j}=2\P_{\frac{i+j-1}{2}}$,
  so \[ \{a+b : (a,b)\in G\} \subseteq \{2\P_1,2\P_2,\ldots,2\P_{n-1}\}.\]
  This lets us use \cref{lem:tao} for $(A,B,G)$ with $N=n-1$.
  We obtain $|\{a-b : (a,b)\in G\}|\le (n-1)^{11/6}$. However,
  since two Abelian squares $w[i..j]$ and $w[i'..j']$ are equivalent if and only if $\P_{j}-\P_{i-1}=\P_{j'}-\P_{i'-1}$,
  we actually have \[\SQPABEL(w)=|\{a-b : (a,b)\in G\}| \le (n-1)^{11/6}.\]
  Since the choice of $w$ was arbitrary, we conclude $\SQPABEL(n,\sigma)\le (n-1)^{11/6}$.
  \end{proof}
  
  \begin{remark}
    Katz \& Tao~\cite{katz1999bounds} apply a construction of Ruzsa \cite{ruzsa} to show
    that the upper bound on $|\{a-b : (a,b)\in G\}|$ cannot be improved beyond $N^{\log_3(6)}\approx N^{1.631}$.
    However, their example does not seem to adapt to the setting of Abelian squares.
  \end{remark}
 
  In the second part of this section we show that a large number of
  Abelian squares enforces that a word contains a large number of blocks.
  Recall that a block is a maximal uniform fragment, i.e.,
  a maximal fragment whose letters are all equal.
  
  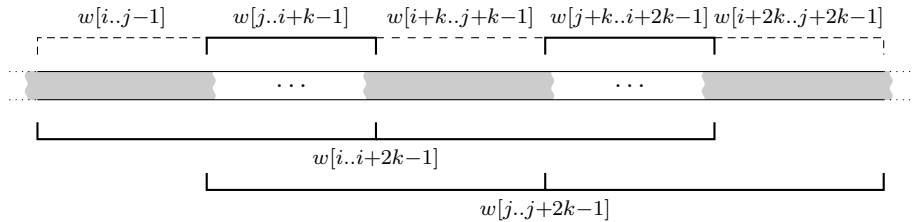
\begin{figure}[b]
 \centering
 \begin{tikzpicture}[scale=0.75]

\begin{scope}
\clip (-1,0) rectangle (15,0.5);
\fill[decorate,decoration = {snake,amplitude =.2mm, segment length = 2mm},fill=black!20] (5.3,-1) rectangle (8.7,1);
\fill[decorate,decoration = {snake,amplitude =.2mm, segment length = 2mm},fill=black!20] (-.7,-1) rectangle (2.7,1);
\fill[decorate,decoration = {snake,amplitude =.2mm, segment length = 2mm},fill=black!20] (11.3,-1) rectangle (14.7,1);
\end{scope}
 \draw (4,0.25) node {$\ldots$};
 \draw (10,0.25) node {$\ldots$};
 \draw (-.5,0) -- (14.5,0) (-.5,.5) -- (14.5,.5);
 \draw[dotted] (-1,0) -- (-.5,0) (14.5,0) -- (15,0) (-1,.5) -- (-.5,.5) (14.5,.5) -- (15,.5);

  \draw[yshift=0.8cm,thick] (2.5,0) -- (2.5,0.3) -- (5.5,0.3) -- (5.5,0);
  \draw[xshift=3cm,yshift=0.8cm,dashed] (2.5,0.3) -- (5.5,0.3);
  \draw[xshift=6cm,yshift=0.8cm,thick] (2.5,0) -- (2.5,0.3) -- (5.5,0.3) -- (5.5,0);
  
  \draw (1,1.1) node[above] {\footnotesize{$w[i..j{-}1]$}};
  \draw (4,1.1) node[above] {\footnotesize{$w[j..i{+}k{-}1]$}};  
  \draw (7,1.1) node[above] {\footnotesize{$w[i{+}k..j{+}k{-}1]$}};
  \draw (10,1.1) node[above] {\footnotesize{$w[j{+}k..i{+}2k{-}1]$}};  
  \draw (13,1.1) node[above] {\footnotesize{$w[i{+}2k..j{+}2k{-}1]$}};

  \draw[densely dashed,yshift=0.8cm] (-0.5,0) -- (-0.5,0.3) -- (2.5,0.3);
  \draw[densely dashed,yshift=0.8cm,xshift=12cm] (-0.5,0.3) -- (2.5,0.3) -- (2.5,0);

  \begin{scope}[yshift=.6cm]
  \draw[yshift=-1cm,thick] (-0.5,0) -- (-0.5,-0.3) -- (5.5,-0.3) -- (5.5,0);
  \draw[xshift=6cm,yshift=-1cm,thick] (-0.5,-0.3) -- (5.5,-0.3) -- (5.5,0);
  \draw (5.5,-1.3) node[below] {\footnotesize{$w[i..i{+}2k{-}1]$}};
\end{scope}
  \begin{scope}[xshift=3cm,yshift=-.3cm]
  \draw[yshift=-1cm,thick] (-0.5,0) -- (-0.5,-0.3) -- (5.5,-0.3) -- (5.5,0);
  \draw[xshift=6cm,yshift=-1cm,thick] (-0.5,-0.3) -- (5.5,-0.3) -- (5.5,0);
  \draw (5.5,-1.3) node[below] {\footnotesize{$w[j..j{+}2k{-}1]$}};
  \end{scope}

\end{tikzpicture}
 \caption{
   Illustration of the ``$j-i < k$'' case in the proof of \cref{lem:uniform}.
   Dark fragments represent occurrences of the same letter $a$.
   In the proof we show that the dashed fragments all have the form $a^{j-i}$, hence the two Abelian squares
   in the bottom have the same Parikh vectors.
 }\label{fig:uni}
\end{figure}

  \begin{lemma}\label{lem:uniform}
  Let $w$ be a word, $k$ be a positive integer and let $i,j$, $i<j$, be indices such that $w[i..i+2k-1]$ and $w[j..j+2k-1]$
  are Abelian squares. If $w[i+k..j+k-1]$ is uniform, then $\P(w[i..i+2k-1])=\P(w[j..j+2k-1])$.
  \end{lemma}
  \begin{proof}
  Let us define $a\in \Sigma$ so that $w[i+k..j+k-1]=a^{j-i}$.
  First, we suppose that $j-i\ge k$. 
  In this case 
  \[w[i+k..i+2k-1]=w[j..j+k-1]=a^k,\]
  and thus $w[i..i+2k-1]=w[j..j+2k-1]=a^{2k}$.
  
  Consequently, we may assume that $j-i < k$; see \cref{fig:uni}. As $w[i..i+2k-1]$ and $w[j..j+2k-1]$ are Abelian squares, we have
  \[\P(w[i..i+k-1])+\P(w[j+k..j+2k-1])=\P(w[j..j+k-1])+\P(w[i+k..i+2k-1]).\]
  Since $w[i..i+k-1]$ overlaps with $w[j..j+k-1]$ on $w[j..i+k-1]$, and $w[i+k..i+2k-1]$ overlaps with $w[j+k..j+2k-1]$ on $w[j+k..i+2k-1]$
  (see the upper part of \cref{fig:uni}), this yields:
  \[ \P(w[i..j-1])+\P(w[i+2k..j+2k-1])=\P(w[i+k..j+k-1])+\P(w[i+k..j+k-1]).\]
  Hence, due to $w[i+k..j+k-1]=a^{j-i}$, we must also have $w[i..j-1]=a^{j-i}$ and $w[i+2k..j+2k-1]=a^{j-i}$.
  Because
  \[\P(w[j..j+2k-1])=\P(w[i..i+2k-1])-\P(w[i..j-1])+\P(w[i+2k..j+2k-1])\]
  this concludes the proof.  
  \end{proof}
  
  \begin{theorem}\label{thm:ab-upper-bound-k-blocks}
    A word $w$ of length $n$ with $m$ blocks contains at most $m$ non\-equivalent Abelian squares of any fixed length. 
    Consequently, $\SQPABEL(w) \le \tfrac{nm}{2}.$
  \end{theorem}
  \begin{proof}
    Suppose that there are $s$ nonequivalent Abelian squares of length $2k$.  
    Let us fix an arbitrary occurrence of each square and let $i_1<\dots < i_s$ be their starting positions. 
    By \cref{lem:uniform}, none of the words $w[i_p+k..i_{p+1}+k-1]$, $1\le p < s$, is uniform. 
    Consequently, $w$ has at least $s$ blocks, i.e., $s \le m$.
  \end{proof}

  \cref{thm:ab-upper-bound-k-blocks} in particular implies a tight asymptotic bound for the
  number of non-equivalent Abelian squares in the lower-bound family of words $\ww_k$.
  \begin{observation}
    $\SQPABEL(\ww_k) = \Theta(|\ww_k|^{1.5})$.
  \end{observation}

  \section{Bounds for Order-Preserving Squares}\label{sec:op}
  Recall that $uv$ is an order-preserving square if $|u|=|v|$ and there exists a strictly increasing
  bijection $f : \Alph(u) \to \Alph(v)$ such that $v[i]=f(u[i])$ for all $i=1,\ldots,|u|$.

  \begin{remark}\label{rmk:op_last}
    A known property of ordinary squares is that each position of a word contains at most
    two rightmost occurrences of a square; see \cite{fraenkel-simpson}.
    This property immediately implies that a word of length $n$ contains at most $2n$ distinct squares.
    Unfortunately, for order-preserving squares an analogous property does not hold.
    For example, the following word of length 28 being a permutation of $\{0,\ldots,27\}$:
    \[0\,\, 3\,\, 1\,\, 6\,\, 2\,\, 7\,\, 4\,\, 8\,\, 5\,\, 11\,\, 9\,\, 13\,\, 10\,\, 16\,\, 12\,\, 17\,\, 14\,\, 20\,\, 15\,\, 21\,\, 18\,\, 22\,\, 19\,\, 25\,\, 23\,\, 26\,\, 24\,\, 27\]
    contains three rightmost occurrences of nonequivalent order-preserving squares (of lengths 16, 20, and 28) starting at the first position.
  \end{remark}
  
  Recall that $uv$ is a parameterized square
  if $|u|=|v|$ and there is a bijection $f: \Alph(u)\to \Alph(v)$ such that
  $v[i]=f(u[i])$ for all $i=1,\ldots,|u|$.
  Note that, obviously, every order-preserving square is a parameterized square.
  A parameterized square $uv$ is called \emph{imbalanced} if $\Alph(u) \ne \Alph(v)$.

  \begin{lemma}\label{lem:imbalanced}
  Let $w\in \Sigma^n$ and $\sigma=|\Sigma|$.
  At most $\binom{\sigma}{2} n$ fragments of $w$ are imbalanced parameterized squares.
  \end{lemma}
  \begin{proof}
    It suffices to prove that at most $\binom{\sigma}{2}$ prefixes of $w$ are imbalanced parameterized squares.
  We shall construct an injective function $g$ mapping such squares to 2-element subsets of $\Sigma$:
  we define $g(uv)=\{a,b\}$ where $b$ is leftmost letter in $v$ which does not belong to $u$, and $a$ is its counterpart in $u$,
  i.e., $a=f^{-1}(b)$ where $f$ is the bijection corresponding to $uv$.
  Let $i$ be the leftmost position such that $u[i]=a$ and $v[i]=b$.
  Observe that $w[i]=a$ and $w[i+|u|]=b$
  are the leftmost occurrences of $a$ and $b$, respectively, in $w$. Consequently, $|u|$ can be reconstructed
  as the difference between these two positions. Hence, $g$ is indeed an injection.
  \end{proof}
  
  \begin{corollary}\label{cor:op}
  Let $w\in \Sigma^n$ and $\sigma=|\Sigma|$. At most $\binom{\sigma}{2} n$ fragments of $w$
  are order-preserving squares but not ordinary squares.
  \end{corollary}
  \begin{proof}
  Let $w[i..j]=uv$ be an order-preserving square. If $\Alph(u)\ne \Alph(v)$, then $w[i..j]$ is an imbalanced
  parameterized square. Otherwise, the corresponding monotone bijection $f: \Alph(u)\to \Alph(v)$ must be the identity.
  Hence, $w[i..j]$ is an ordinary square. Consequently, \cref{lem:imbalanced} concludes the proof.
  \end{proof}
  
  \begin{theorem}\label{thm:op}
  $\SQOP(n,\sigma)\le (\binom{\sigma}{2}+\frac{11}{6})n$.
  \end{theorem}
  \begin{proof}
  A result of Deza et al.~\cite{DBLP:journals/dam/DezaFT15} shows that a word of length $n$ contains
  at most $\frac{11}{6}n$ distinct ordinary squares.
  By \cref{cor:op}, the remaining order-preserving squares have at most $\binom{\sigma}{2} n$ occurrences.
  \end{proof}

  \section{Bounds for Parameterized Squares}\label{sec:param}
  We start with a remark similar to \cref{rmk:op_last}.

  \begin{remark}
   The word
    \[0\, 1\, 2\, 0\, 3\, 0\, 1\, 3\, 0\, 2\, 3\, 1\, 3\, 0\]
    contains three parameterized squares starting at the first position:
    \[0\, 1\, 2\, 0\ \ 3\, 0\, 1\, 3,\quad 0\, 1\, 2\, 0\, 3\ \ 0\, 1\, 3\, 0\, 2,\quad 0\, 1\, 2\, 0\, 3\, 0\, 1\ \ 3\, 0\, 2\, 3\, 1\, 3\, 0\]
    such that no parameterized square equivalent to any of these three occurs anywhere else in the word.
  \end{remark}
  
  For a word $w\in \Sigma^*$ let us define $\L(w)\in \Sigma^*$ which results by removing all characters
  except for the \emph{last} occurrence of each letter. Note that in the resulting word
  each character of $w$ occurs exactly once, i.e., $\L(w)$ is a \emph{permutation} of $\Alph(w)$.
  We denote the family of permutations of $\Sigma$ by $S_{\Sigma}$.
  Throughout this section we consider permutations as strings over $\Sigma$, i.e., $S_{\Sigma}\subseteq \Sigma^\sigma$.
  For a permutation $\pi \in S_\Sigma$ and a letter $a\in \Sigma$ we define $\ind_{\pi}(a)$ as the 0-based index of $a$ in $\pi$
  \emph{counting from the right}. We extend $\ind_{\pi}$ to arbitrary words $w$ and characters $a\in \Alph(w)$ setting $\ind_{w}(a)=\ind_{\L(w)}(a)$,
  i.e., $\ind_{w}(a)$ is the number of distinct characters after the last occurrence of $a$ in $w$.
  
  For a permutation $\pi \in S_{\Sigma}$ we define an encoding
  $h_{\pi} : \Sigma^* \rightarrow \{0,\ldots,\sigma-1\}^*$ where $h_\pi(w)$ is a word $z$ of length $|w|$
  such that for $i=1,\ldots,|w|$ we have $z[i]=\ind_{\pi w[1..i-1]}(w[i])$.
  Intuitively, $h_\pi(w)$ shows, for each position $i$ of $w$, how many distinct letters are there between
  $w[i]$ and the previous occurrence of the letter $w[i]$ in $w$.
  However, if $w[i]$ occurs for the first time at position $i$, $h_\pi$ uses the word $\pi$
  that is prepended to $w$ to determine the previous occurrence.
  
\begin{example}
We have $\L(abcba)=cba$ and $\L(ababb)=ab$.
For $\pi = abc$ we have
$h_\pi(abcba)=22212$; see Table~\ref{tab:1}. We also have $h_\pi(ababb)=22110$.

\begin{table}[h]
  \begin{center}
    \caption{Intermediate steps of the computation of $h_{\pi}(ababb)=22212$ for $\pi=abc$.}\label{tab:1}
    \footnotesize  
  \begin{tabular}{c|c|c|c|c}
    $i$ & $\pi w[1..i-1]$ & $\L(\pi w[1..i-1])$ & $w[i]$ & $z[i]$\\\hline
      1 & $abc$ & $abc$ & $a$ & 2 \\
      2 & $abca$ & $bca$ & $b$ & 2 \\
      3 & $abcab$ & $cab$ & $c$ & 2 \\
      4 & $abcabc$ & $abc$ & $b$ & 1 \\
      5 & $abcabcb$ & $acb$ & $a$ & 2 
  \end{tabular}
\end{center}
\end{table}
\end{example}

Below, we relate parameterized square prefixes of a word $w$ with ordinary square prefixes of its encodings $h_\pi(w)$.
More precisely, we show that $w$
has a parameterized square prefix of a given length if and only if there exists a permutation $\pi\in S_\Sigma$ such that $h_\pi(w)$ has an (ordinary) square prefix of the same length. We start by listing a few simple properties of the notions introduced above.

\begin{observation}\label{obs:ord}
For every words $u,v,w\in \Sigma^*$ and every bijection $f:\Sigma\to \Sigma$ (extended to a morphism $f:\Sigma^*\to \Sigma^*$), we
have
\begin{enumerate}[(i)]
  \item $\L(uwvw)=\L(uvw)$,
  \item $\L(uvw)=\L(u\L(v)w)$,
  \item $f(\L(u))=\L(f(u))$,
  \item $h_{\pi}(u)=h_{f(\pi)}(f(u))$.
\end{enumerate}
\end{observation}

\begin{lemma}\label{lem:eq}
Let $w\in \Sigma^n$, $\pi \in S_\Sigma$, and $z=h_\pi(w)$.
If $z[i..j]=z[i'..j']$, then $w[i..j]$ and $w[i'..j']$ are equivalent in the parameterized sense.
\end{lemma}
\begin{proof}
We proceed by induction on the length of the fragments. For length 0 the claim is trivial.
Thus, suppose that it holds for all lengths strictly smaller than $j-i+1$. 
Consequently, $w[i..j-1]$ and $w[i'..j'-1]$ are parameterized equivalent with some witness bijection $f:\Alph(w[i..j-1])\to \Alph(w[i'..j'-1])$.

We consider two cases. First, suppose that
\[z[j]=z[j']\ge|\Alph(w[i..j-1])|=|\Alph(w[i'..j'-1])|.\]
By definition of $h_{\pi}$ this means that $w[j]\notin \Alph(w[i..j-1])$ and $w[j']\notin \Alph(w[i'..j'-1])$.
Consequently, $f$ can be extended with $w[j]\mapsto w[j']$, which yields the witness bijection for equivalence of $w[i..j]$
and $w[i'..j']$.

Next, suppose that
\[z[j]=z[j'] = k < |\Alph(w[i..j-1])|=|\Alph(w[i'..j'-1])|.\]
This means that $w[j]$ is the $k$-th element of $\Alph(w[i..j-1])$ ordered according to the last occurrence in $w[i..j-1]$.
Similarly, $w[j']$ is the $k$-th element of $\Alph(w[i'..j'-1])$ ordered in the same way with respect to $w[i'..j'-1]$.
Since $w[i..j-1]$ and $w[i'..j'-1]$ are parameterized equivalent, the relative positions of these last occurrences are the same.
Hence, $f(w[j])=f(w[j'])$ and $f$ is the witness bijection of equivalence between $w[i..j]$ and $w[i'..j']$.
\end{proof}

\begin{lemma}\label{lem:comp}
Let $v\in \Sigma^*$, $\pi \in S_{\Sigma}$, and let $\pi \odot v = \L(\pi v)$.
For every $w\in \Sigma^*$ we have $h_{\pi}(vw) = h_\pi(v)h_{\pi \odot v}(w)$.
\end{lemma}
\begin{proof}
Let us consider a position $i$ of $vw$. If $i\le |v|$, we clearly have 
\[(h_{\pi}(vw))[i]=\ind_{\pi (vw)[1..i-1]}((vw)[i])=\ind_{\pi v[1..i-1]}(v[i])=(h_{\pi}(v))[i]\]
since $v[1..i]=(vw)[1..i]$. Thus, let us consider $i > |v|$.
Then, we have
\begin{align*}
  (h_{\pi}(vw))[i] &= \ind_{\pi(vw)[1..i-1]}((vw)[i])=\ind_{\pi vw[1..i-|v|-1]}(w[i-|v|])=\\
                   &= \ind_{(\pi \odot v) w[1..i-|v|-1]}(w[i-|v|]) = (h_{\pi \odot v}(w))[i-|v|]
\end{align*}
since \[\L(\pi vw[1..i-|v|-1])=\L(\L(\pi v)w[1..i-|v|-1])=\L((\pi \odot v) w[1..i-|v|-1])\] by \cref{obs:ord}(ii).
\end{proof}

\begin{lemma}\label{lem:pi}
Let $w\in \Sigma^*$ be a parameterized square. There exists $\pi \in S_\Sigma$ such that $h_{\pi}(w)$
is an ordinary square. 
\end{lemma}
\begin{proof}
Let $w=uv$ be the decomposition into halves and let $f:\Alph(u)\to \Alph(v)$ be the witness bijection
of the parameterized equivalence of $u$ and $v$.
By \cref{lem:comp}, we have $h_{\pi}(w)=h_{\pi}(u)h_{\pi \odot u}(v)$.
We shall choose $\pi$ so that $h_{\pi}(u)=h_{\pi \odot u}(v)$. 

Let us extend $f$ to a bijection
$f : \Sigma \to \Sigma$ using identity on $\Sigma\setminus (\Sigma_u\cup \Sigma_v)$ and
an arbitrary bijection $(\Sigma_v \setminus \Sigma_u)\to (\Sigma_u\setminus \Sigma_v))$
where $\Sigma_u = \Alph(u)$ and $\Sigma_v=\Alph(v)$.

Let $r$ be the rank of $f$, i.e., the smallest positive integer such that $f^r = \id$,
and let $\rho\in S_{\Sigma\setminus (\Alph(u)\cup \Alph(v))}$ be an arbitrary permutation.
We claim that $\pi = \L(\rho f^{0}(u)f^{1}(u)\cdots f^{r-1}(u))$ is a permutation of $\Sigma$ satisfying $\pi \odot u = f(\pi)$.
First, note that $f^0(u)=u$ and $f^1(u)=v$, so $\pi\in S_{\Sigma}$.
Next, we apply \cref{obs:ord}:
\begin{multline*}
  \pi \odot u = \L(\pi u) = \L(\L(\rho f^{0}(u)f^{1}(u)\cdots f^{r-1}(u))u)
               \,\stackrel{\mathclap{\mbox{\tiny\ref{obs:ord}(ii)}}}{=}\, \L(\rho uf^{1}(u)\cdots f^{r-1}(u)u)\\
               \,\stackrel{\mathclap{\mbox{\tiny\ref{obs:ord}(i)}}}{=}\, \L(\rho f^{1}(u)\cdots f^{r-1}(u)u)
               = \L(f(\rho) f^{1}(u)\cdots f^{r-1}(u)f^r(u))\\
               = \L(f(\rho f^{0}(u)f^{1}(u)\cdots f^{r-1}(u)))
               \,\stackrel{\mathclap{\mbox{\tiny\ref{obs:ord}(iii)}}}{=}\, f(\L(\rho f^{0}(u)f^{1}(u)\cdots f^{r-1}(u)))
               = f(\pi).
\end{multline*}
Finally, using \cref{obs:ord}(iv), we conclude that $h_{\pi}(u)=h_{\pi \odot u}(v)$ since $v=f(u)$ and $\pi \odot u = f(\pi)$.
\end{proof}
 
 Next, we shall apply the following standard result to prove its counterpart for parameterized equivalence.
  \begin{fact}[\cite{fraenkel-simpson}]\label{fct:fraenkel}
  A word $w\in \Sigma^*$ has at most two prefixes which are ordinary squares without another occurrence in $w$.
  \end{fact}
  
\begin{lemma}\label{lem:factorial}
A word $w\in \Sigma^*$ has at most $2\sigma!$ prefixes which are parameterized squares
without another (parameterized) occurrence in $w$.
\end{lemma}
\begin{proof}
By \cref{lem:pi} for every prefix of $w$ being a parameterized square there is a permutation $\pi\in S_{\Sigma}$
such that the corresponding prefix of $h_\pi(w)$ is an ordinary square. Fact~\ref{fct:fraenkel} implies
that for a fixed $\pi$ at most two such ordinary squares do not occur later in $h_{\pi(w)}$. However,
by \cref{lem:eq}, such a later occurrence in $h_{\pi(w)}$ means that the prefix of $w$ has
another (parameterized) occurrence in $w$. Combining these results yields an upper bound of $2|S_{\Sigma}|=2\sigma!$
on the number of prefixes being parameterized squares
without another (parameterized) occurrence in $w$.
\end{proof}

\Cref{lem:factorial} immediately yields a bound for the parameterized squares.

\begin{theorem}
For every positive integers $n$ and $\sigma$ we have 
    $\SQPPARAM(n,\sigma) \le 2\sigma!n$ 
    and $\SQPARAM(n,\sigma)\le 2(\sigma!)^2n$.
\end{theorem} 
\begin{proof}
First, let us count parameterized squares up to equivalence.
\Cref{lem:factorial} implies that at most $2\sigma!$ parameterized squares have their leftmost occurrence
at any given position, which gives $2\sigma!n$ non-equivalent parameterized squares in total.

To count parameterized squares up to equality of subwords, it suffices to observe that
every class of parameterized equivalence has at most $\sigma!$ elements (the class $[u]$ has exactly $|\Alph(u)|!$ elements).
Hence, the number of parameterized squares distinct as subwords is at most $2(\sigma!)^2n$. 
\end{proof}

\section{Infinite Words Avoiding Nonstandard Squares and Cubes}\label{sec:infinite}
  Let us recall that there exist infinite ternary words that avoid ordinary squares \cite{Thue}
  (and obviously there is no such binary word).
  It is also known that there are infinite words over a 4-letter alphabet avoiding Abelian squares
  while over 3-letter alphabets such words do not exist \cite{DBLP:conf/icalp/Keranen92}.
  Here, we investigate an analogous problem for other nonstandard repetitions.

\subsection{Avoiding Order-Preserving Squares}
  We say that a word is op-square-free if it does not contain an order-preserving square
  of length greater than 2.
  Let $\Sigma_3=\{0,1,2\}$ ordered in the natural way. Consider the morphism:
    \[ \lleft\ \;:\; 0\mapsto 10,\, 1\mapsto 11,\, 2\mapsto 12.\]
  It satisfies the following property, which lets us construct an op-square-free word. 

    \begin{observation}\label{obs:simple}
      For every symbols $a,b,c \in \Sigma_3$ we have:
      \begin{enumerate}[(i)]
        \item $1\,a\approx 1\,b\: \Leftrightarrow\: a=b$;
        \item $a\,1\,b\approx 1\,c\,1 \:\Rightarrow\: a=b$.
      \end{enumerate}
    \end{observation}

  \begin{lemma}\label{lem:left}
    If a word $w\in \Sigma_3^*$ is square-free, then $\lleft(w)$ is op-square-free.
  \end{lemma}

  \begin{proof}
    Let $\approx$ denote the order-preserving equivalence
    (i.e., $u \approx v$ if $|u|=|v|$ and $uv$ is an order-preserving square).

    \noindent
    Suppose to the contrary that $w'=\lleft(w)$ contains an order-preserving square
    $u'v'=w'[i..i+2k-1]$, with $|u'|=|v'|=k \ge 2$.
    We consider four cases depending on the parity of $i$ and $k$.

    If $2 \mid k$ and $2 \nmid i$, then $u'$ and $v'$ start with a 1 and every second symbol of each of them is a 1.
    Consequently, by \cref{obs:simple}(i), $u'=v'$.
    Moreover, in this case we have $u'=\lleft(u)$ and $v'=\lleft(v)$ for some subword $uv$ of $w$.
    Hence, $uv$ is a square in $w$, a contradiction.

    If $2 \mid k$ and $2 \mid i$, then $w'[i-1...i+2k-2]$ is also an order-preserving square.
    The conclusion follows from the previous case.

    If $2 \nmid k$ and $2 \nmid i$, then $u'$ and $v'$ start with $1c1$ and $a1b$ for some
    $a,b,c \in \Sigma_3$, respectively.
    By \cref{obs:simple}(ii), we conclude that $a=b$, which implies a square $ab$ in $w$, a contradiction.

    The final case, $2 \nmid k$ and $2 \mid i$, also implies a 2-letter square in $w$
    just as in the previous case.
    This completes the proof that $w'$ is op-square-free.
  \end{proof}

  \noindent
  We apply \cref{lem:left} to all prefixes of an infinite square-free word \cite{Thue} over a ternary alphabet
  and obtain the following result.

  \begin{theorem}
    There exists an infinite op-square-free word over 3-letter alphabet.
  \end{theorem}

  \begin{example}
    If we apply the morphism $0\mapsto 10,\, 1\mapsto 11,\, 2\mapsto 12$ to the infinite square-free word that starts with
    \begin{center}
      012021012102012021\,$\cdots$
    \end{center}
    we obtain an op-square-free word that starts with:
    \begin{center}
      10\,11\,12\,10\,12\,11\,10\,11\,12\,11\,10\,12\,10\,11\,12\,10\,12\,11\,$\cdots$
    \end{center}
  \end{example}

\subsection{Avoiding Parameterized Cubes}
  A parameterized cube is a word $uvw$ such that both $uv$ and $vw$ are parameterized squares.
  A word is called parameterized-square-free (parameterized-cube-free)
  if it does not contain parameterized squares (parameterized cubes) of length
  greater than 3.
  We show that there is no infinite parameterized-square-free word,
  and we construct a binary parameterized-cube-free word.

  \begin{theorem}
    There is no infinite parameterized-square-free word.
  \end{theorem}
  \begin{proof}
    Suppose to the contrary that such an infinite word $x$ exists.
    In the proof we denote symbols of $\Alph(x)$ by $a,b,c,d$.
    Note that every suffix of $x$ has to contain two adjacent equal symbols.
    This is because $abcd$ for $a \ne b$ and $c \ne d$ is a parameterized square.
    Moreover, $x$ has to contain some three adjacent equal symbols.
    The reason is that $abbd$ for $a \ne b \ne d$ is a parameterized square.

    We can therefore assume that $x$ contains a subword $aaa$.
    To avoid a parameterized square of length 4, this subword must be followed in $x$
    by some letter $b \ne a$.
    For the same reason the next letter $c$ must satisfy $c \ne b$, and afterwards
    the subword $aaabc$ must be followed by two more occurrences of $c$.
    Finally the next letter must be $d \ne c$ to avoid a parameterized square $cccc$.
    We conclude that $x$ contains a subword $aaabcccd$ for $b \ne a$ and $d \ne c$,
    which turns out to be a parameterized square.
    This contradiction completes the proof.
  \end{proof}

  We proceed to a construction of a binary parameterized-cube-free word.
  An \emph{antisquare} is a nonempty word of the form $x\bar{x}$, where $\bar{x}$ denotes bitwise negation of $x$.
  For example, 011 100 is an antisquare.
  In the proof we will use the following characterisation of binary parameterized squares.
  
  \begin{observation}\label{obs:anti}
    For binary alphabet each parameterized square is an ordinary square or an antisquare.
  \end{observation}

  Let $\tau$ be the infinite Thue-Morse word.
  Recall that $\tau$ is cube-free \cite{Thue2}.
  Also recall the morphism $\lleft$ defined just before \cref{lem:left}
  (here we consider only $\lleft(0)$ and $\lleft(1)$).

  \begin{theorem}
    The word $\lleft(\tau)$ is parameterized-cube-free.
  \end{theorem}
  \begin{proof}
    Suppose to the contrary that $u_1u_2u_3$ is a parameterized cube in $\lleft(\tau)$, with
    $|u_1|=|u_2|=|u_3|=k>1$.
    Note that $\lleft(\tau)$ does not contain 6 ones in a row.
    Hence, at least one of the words $u_1,u_2,u_3$ contains 0, therefore each of them contains 0.
    Moreover, every second symbol of $u_1,u_2,u_3$ is 1.

    Recall from \cref{obs:anti} that a binary parameterized square is either
    an ordinary square or an antisquare.
    If $2 \mid k$, then the ones of every second position of $u_1,u_2,u_3$ align
    and $u_1u_2$, $u_2u_3$ must be ordinary squares.
    Therefore $u_1u_2u_3$ is an ordinary cube in $\lleft(\tau)$ which induces a cube in $\tau$.

    If $2 \nmid k$, the same argument implies that both $u_1u_2$ and $u_2u_3$ are antisquares.
    Because of the ones on every second position of $u_1,u_2,u_3$
    we actually have $u_1=0101\cdots$, $u_2=1010\cdots$, $u_3=0101\cdots$ or
    $u_1=1010\cdots$, $u_2=0101\cdots$, $u_3=1010\cdots$.
    In both cases we obtain a cube $(10)^3$ in $\lleft(\tau)$ which induces  $0^3$ in $\tau$.
  \end{proof}

  \begin{example}
    The first few symbols of the Thue-Morse word $\tau$ are:
    \begin{center}
      011010011001011010\,$\cdots$
    \end{center}
    We apply the morphism $0\mapsto 10,\, 1\mapsto 11$ to obtain a parameterized-cube-free word
    starting with:
    \begin{center}
      10\,11\,11\,10\,11\,10\,10\,11\,11\,10\,10\,11\,10\,11\,11\,10\,11\,10\,$\cdots$
    \end{center}
  \end{example}

  \section{Final Remarks}
  We have presented several combinatorial results related to the maximum number of
  nonstandard squares in a word of length $n$.
  For Abelian squares we have shown that for $\sigma\ge 2$:
  \[\SQABEL(n,\sigma)=\Theta(n^2),\ \SQPABEL(n,\sigma)=\Omega(n^{3/2}),\ \mbox{and}\ \SQPABEL(n,\sigma) = O(n^{11/6}).\]

  For squares in order-preserving and parameterized setting we have shown
  that their maximum number is linear of $n$ for a constant alphabet.
  We have also presented examples of infinite words over a minimal alphabet that avoid squares in
  order-preserving setting and cubes in parameterized setting, respectively.

  The main open question that arises from our work is to provide an upper bound for $\SQPABEL(n,2)$.
  We have made a step towards this bound by showing that the maximum number of distinct Abelian squares
  in a word of length $n$ containing $m$ blocks is $O(nm)$.
  The remaining open questions are connected to $\SQPOP(n,\sigma)$ and $\SQPPARAM(n,\sigma)$
  for arbitrary $\sigma$ (not necessarily a constant).
  Based on experimental results, we state the following hypothesis:
\begin{hypothesis}
   For every $\sigma \ge 2$,
  $\SQPABEL(n,\sigma) = \Theta(n^{3/2})$, $\SQPOP(n,\sigma)=\Theta(n)$, and $\SQPPARAM(n,\sigma)=\Theta(n)$
  (with constant factors independent of $\sigma$).
\end{hypothesis}

\bibliographystyle{plainurl}
\bibliography{squares}

\end{document}